\documentclass{article}
\usepackage{arxiv}
\usepackage[utf8]{inputenc} 
\usepackage[T1]{fontenc}    
\usepackage{url}            
\usepackage{booktabs}       
\usepackage{hyperref}
\usepackage{amsfonts}       
\usepackage{nicefrac}       
\usepackage{microtype}      
\usepackage{lipsum}
\usepackage{graphicx}
\usepackage{subfiles}
\usepackage{amsmath}
\usepackage{subfig}
\usepackage{refcount}
\usepackage{algorithmic}
\usepackage{algorithm}

\usepackage{cite}
\usepackage{tikz}
\usepackage[english]{babel}
\usepackage{amsthm}
\newtheorem{theorem}{Theorem}
\usepackage{listings}
\usepackage{xcolor}
\usepackage[export]{adjustbox}
\definecolor{mygray}{rgb}{0.5,0.5,0.5}
\title{Optimality Study of Existing Quantum Computing Layout Synthesis Tools}

\author{
 Bochen Tan\\
  CS, UCLA\\
  \texttt{bctan@cs.ucla.edu}\\
   \and
 \textbf{Jason Cong} \\
  CS, UCLA\\
}

\begin{document}
© 2020 IEEE.  Personal use of this material is permitted.  Permission from IEEE must be obtained for all other uses, in any current or future media, including reprinting/republishing this material for advertising or promotional purposes, creating new collective works, for resale or redistribution to servers or lists, or reuse of any copyrighted component of this work in other works.

\bigskip
\bigskip
\bigskip
\maketitle
\begin{abstract}
Layout synthesis, an important step in quantum computing, processes quantum circuits to satisfy device layout constraints. 
In this paper, we construct QUEKO benchmarks for this problem, which have known optimal depths and gate counts. 
We use QUEKO to evaluate the optimality of current layout synthesis tools, including Cirq from Google, Qiskit from IBM, $\mathsf{t}|\mathsf{ket}\rangle$ from Cambridge Quantum Computing, and recent academic work. 
To our surprise, despite over a decade of research and development by academia and industry on compilation and synthesis for quantum circuits, we are still able to demonstrate large optimality gaps: 1.5-12x on average on a smaller device and 5-45x on average on a larger device.
This suggests substantial room for improvement of the efficiency of quantum computer by better layout synthesis tools.
Finally, we also prove the NP-completeness of the layout synthesis problem for quantum computing.
We have made the QUEKO benchmarks open-source.

\end{abstract}

\section{Introduction}\label{sec:introduction}

Recently, a quantum processor ``Sycamore'' from Google was shown to have a clear advantage over classical supercomputers on a problem named sampling random quantum circuit~\cite{arute-2019-quantum}. 
It is widely expected that in the near future, quantum computing (QC) will outperform its classical counterpart solving more and more problems. 
Achieving this computational advantage, however, requires executing larger and larger quantum circuits.

A quantum circuit consists of quantum gates acting on qubits. 
It was shown that only gates acting on one or two qubits are required for universal quantum computing~\cite{nielsen-2010-quantum}. 
After a quantum circuit is designed, it needs to be mapped to a QC device. 
However, qubit connections required by two-qubit gates are often greatly constrained by QC device layouts. 
QC layout synthesis resolves this issue by producing an initial mapping from the qubits in the circuit to the physical qubits on the QC device, adjusting the mapping to legalize two-qubit gates by inserting some new gates, and scheduling all the gates. 
The resulting circuit preserves the original functionality and is executable on the QC device.

Quantum circuits in earlier experiments used to be only dozens of gates on small devices, e.g., those with 5 or fewer qubits. 
In those cases, layout synthesis was usually realized by exhaustive enumeration. 
However, the task is increasingly intractable as the circuits get deeper and wider. 
Nowadays, a cutting-edge QC experiment requires the execution of a circuit of 53 qubits, 1113 single-qubit gates, and 430 two-qubit gates~\cite{arute-2019-quantum}. 
For a general circuit of this size, the number of possible initial mappings is $53!$, and the subsequent scheduling and legalization steps have large solution space as well. 
Clearly, design automation is necessary. 
In addition, the size of the circuits that QC hardware is able to execute has been scaling exponentially in the past few years~\cite{gambetta-2019-cramming}.
The fast increase of hardware capacity presents an even bigger challenge to layout synthesis.

Several layout synthesis tools are available and there are also benchmarks that help us to compare them. 
However, it is currently unknown how far these tools are away from the optimal solutions.
In this paper, we present QUEKO benchmarks (quantum mapping examples with known optimal) which are quantum circuits with known optimal depth for the given QC device. 
Then, we evaluate four existing QC layout synthesis tools with QUEKO, namely $\mathsf{t}|\mathsf{ket}\rangle$\footnote{\label{nt:tket}\url{https://github.com/CQCL/pytket}}, \verb|greedy| router included in Cirq\footnote{\label{nt:cirq}\url{https://github.com/quantumlib/Cirq}}, \verb|DenseLayout| plus \verb|StochasticSwap| included in Qiskit\footnote{\label{nt:qiskit}\url{https://Qiskit.org/}}, and the recent academic work by Zulehner et al.~\cite{zulehner-2018-efficient}. 
To our surprise, rather large optimality gaps are discovered even for feasible-depth circuits: 1.5-12x on average on a smaller device and 5-45x on average on a larger device.

This result is somewhat surprising as, in comparison, the result of the VLSI circuit placement optimality study conducted more than 15 years ago using the PEKO benchmarks~\cite{chang-2004-optimality} revealed that the optimality gaps of the leading academic and industrial placers were 1.43-2.12X on average for circuits with over two million placable objects, while the quantum circuits used in this study have only up to 54 qubits and about 35 thousand quantum gates, yet shown a much large optimality gap.

The optimality gaps revealed in this study have a strong implication. 
If we can consistently halve the circuit depth by better layout synthesis, we effectively double the decoherence time of a QC device, which is equivalent to a large advancement in experimental physics and electrical engineering. 
Therefore, the gaps call for more research investments into QC layout synthesis. 
To draw a parallel, the VLSI placement optimality study using the PEKO benchmarks~\cite{chang-2004-optimality} spurred further research investment in circuit placement, resulting in wirelength reduction equivalent to two or more generations of Moore's Law scaling, but in a more cost-efficient way~\cite{misc-2007-huge}.

The rest of this paper is organized as follows: Sec.~\ref{sec:background} reviews relevant background of QC, Sec.~\ref{sec:problem} formulates the QC layout synthesis problem; Sec.~\ref{sec:related} reviews related work; Sec.~\ref{sec:approach} provides the construction of QUEKO benchmarks; Sec.~\ref{sec:experiment} evaluates aforementioned tools with QUEKO; Sec.~\ref{sec:complexity} proves the NP-completeness of QC layout synthesis; Sec.~\ref{sec:conclusion} gives conclusion.

\section{Background}\label{sec:background}
\subsection{Qubits}
A qubit is in a quantum state $|\psi\rangle$ represented by a vector in two-dimensional Hilbert space with $L^2$-norm equals $1$
\begin{equation}
|\psi\rangle=\begin{pmatrix} \alpha \\ \beta \end{pmatrix}=\alpha\begin{pmatrix} 1 \\ 0 \end{pmatrix}+ \beta \begin{pmatrix} 0 \\ 1 \end{pmatrix}=\alpha|0\rangle+\beta|1\rangle,
\end{equation}
where the two basis vectors are $|0\rangle$ and $|1\rangle$. A quantum state of multiple qubits lies in the tensor product of individual Hilbert spaces. 
For instance, a general two-qubit state $|\phi\rangle$ is
\begin{equation}
|\phi\rangle= \alpha|0\rangle|0\rangle+\beta|0\rangle|1\rangle+\gamma|1\rangle|0\rangle+\lambda|1\rangle|1\rangle=(\alpha\ \beta\ \gamma\ \lambda)^T,
\end{equation}
where we omit the tensor product notation $\otimes$ between $|\cdot\rangle$s for convenience. 
A joint state of two individual qubits $|\psi_1\rangle\otimes|\psi_0\rangle$ is
\begin{equation}
|\psi_1\rangle|\psi_0\rangle=\begin{pmatrix} \alpha \\ \beta \end{pmatrix}\otimes \begin{pmatrix} \alpha' \\ \beta' \end{pmatrix}=(\alpha\alpha'\ \alpha\beta'\ \beta\alpha'\ \beta\beta')^T.
\end{equation}

\subsection{Quantum Gates} \label{sec:gates}
A quantum gate transforms an input state to an output state. 
For example, some common single-qubit gates are $X$, $H$, and $T$. 
$\dagger$ means transpose complex conjugate.
\begin{equation}
     X|\psi\rangle=\begin{pmatrix} 0 & 1\\ 1 & 0\end{pmatrix}\begin{pmatrix} \alpha \\ \beta \end{pmatrix}=\begin{pmatrix} \beta \\ \alpha \end{pmatrix},\quad H=\frac{\sqrt{2}}{2}\begin{pmatrix} 1 & 1 \\ 1 & -1 \end{pmatrix}, \\
     T=\begin{pmatrix} 1 & 0 \\ 0 & e^{i\pi /4} \end{pmatrix}, \quad T^\dagger=\begin{pmatrix} 1 & 0 \\ 0 & e^{-i\pi /4} \end{pmatrix}.
\end{equation}
Two common two-qubit gates are $CZ$ and $CX$ (also named $CNOT$). 
\begin{equation}
    CZ=\begin{pmatrix} 1 & 0 & 0 & 0\\ 0 & 1 & 0 & 0\\ 0 & 0 & 1 & 0 \\ 0 & 0 & 0 & -1\end{pmatrix},\quad CX=\begin{pmatrix} 1 & 0 & 0 & 0\\ 0 & 1 & 0 & 0\\ 0 & 0 & 0 & 1 \\ 0 & 0 & 1 & 0\end{pmatrix}.
\end{equation}
NAND gates are sufficient for universal classical computing.
For universal quantum computing, there are multiple complete gate sets. 
Table \ref{tab:gates} lists three such sets chosen by QC frameworks Cirq\footnotemark[\getrefnumber{nt:cirq}], Qiskit\footnotemark[\getrefnumber{nt:qiskit}], and pyQuil\footnote{\label{nt:pyquil}\url{https://github.com/rigetti/pyquil}}.
Other gates may be input to these frameworks, but would be translated into a combination of these native gates.
The exact matrix representations of these gates are not specified here because they are irrelevant to the purpose of this paper.
\begin{table}[htb]
\caption{Complete quantum gate set examples}
\begin{center}
\begin{tabular}{|c|c|c|}
\hline
\textbf{Framework} & \textbf{Single-qubit gate}& \textbf{Two-qubit gate} \\
\hline
Cirq (Google) & $Z(\lambda)$, phased $X$ power & $CZ$\\
\hline
Qiskit (IBM) & $U_1(\lambda), U_2(\phi, \lambda)$, $U_3(\theta,\phi,\lambda)$ & $CX$\\
\hline
pyQuil (Rigetti) & $RZ(\lambda), RX(k \pi/2)$  & $CZ$ \\
\hline
\end{tabular}
\label{tab:gates}
\end{center}
\end{table}

Another important gate is $CCX$, or Toffoli, gate which is universal for reversible logic and thus essential for QC logic synthesis~\cite{shende-2003-synthesis}. 
It is a quantum gate on three qubits and can be decomposed into a set of single-qubit and two-qubit gates as shown in the following subsection.
\begin{equation}
CCX=
\begin{pmatrix}
{1} & {0} & {0} & {0} & {0} & {0} & {0} & {0} \\
{0} & {1} & {0} & {0} & {0} & {0} & {0} & {0} \\
{0} & {0} & {1} & {0} & {0} & {0} & {0} & {0} \\
{0} & {0} & {0} & {1} & {0} & {0} & {0} & {0} \\
{0} & {0} & {0} & {0} & {1} & {0} & {0} & {0} \\
{0} & {0} & {0} & {0} & {0} & {1} & {0} & {0} \\
{0} & {0} & {0} & {0} & {0} & {0} & {0} & {1} \\
{0} & {0} & {0} & {0} & {0} & {0} & {1} & {0}
\end{pmatrix}
\end{equation}

An important property between two gates is whether they commute.
If the result of first applying gate 0 followed by gate 1 is the same with the result of the other way around, we say these two gates \textit{commute}.
If two gates act on totally different qubits, they always commute.
For gates acting on a same qubit, it is nontrivial to judge in general.
$T$ and $CZ$ commute because they are both diagonal matrices; but $X$ and $H$ do not commute, since $XH\neq HX$.

For more background knowledge, the reader can refer to a comprehensive textbook on quantum computing and quantum information~\cite{nielsen-2010-quantum}, or a concise textbook on quantum computing~\cite{book07-mermin} with more computer science flavor.

\subsection{Quantum Circuit}

In QC, a circuit or program is usually input as a piece of QASM code~\cite{cross-2017-open}, e.g., Fig.~\ref{fig:qasm}. 
The code is rather simple to read, merely specifying each gate sequentially like instructions in traditional assembly language. 
We thus define a {\it quantum circuit} to be a list of quantum gates $g_1g_2...g_M$.

\begin{figure}[htb]
\centering
\begin{minipage}[t]{0.5\linewidth}
\subfloat[QASM code of Toffoli circuit]{
\begin{lstlisting}
OPENQASM 2.0;
include "qelib1.inc";
qreg q[3]; // Initiate 3 logical qubits
h q[2]; //(*\color{gray} $g_1$*)
cx q[1], q[2]; //(*\color{orange} $g_2$*)
tdg q[2]; //(*\color{gray} $g_3$*)
cx q[0], q[2]; //(*\color{olive} $g_4$*)
t q[2]; //(*\color{gray} $g_5$*)
cx q[1], q[2]; //(*\color{orange} $g_6$*)
tdg q[2]; //(*\color{gray} $g_7$*)
cx q[0], q[2]; //(*\color{olive} $g_8$*)
t q[1]; //(*\color{gray} $g_9$*)
t q[2]; //(*\color{gray} $g_{10}$*)
cx q[0], q[1]; //(*\color{magenta} $g_{11}$*)
h q[2]; //(*\color{gray} $g_{12}$*)
t q[0]; //(*\color{gray} $g_{13}$*)
tdg q[1]; //(*\color{gray} $g_{14}$*)
cx q[0], q[1]; //(*\color{magenta} $g_{15}$*)
\end{lstlisting}
\label{fig:qasm}
}
\end{minipage}
\vspace{\baselineskip}

\begin{minipage}[b]{\linewidth}
\centering
\subfloat[1D diagram of Toffoli circuit]{
\includegraphics[width=0.7\linewidth]{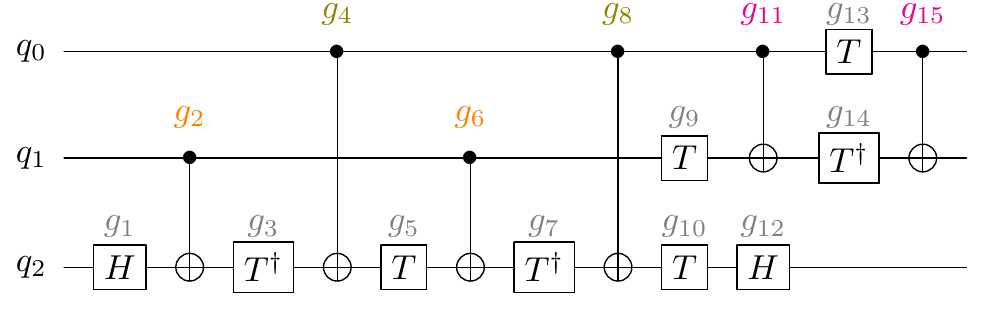}
\label{fig:QCDiagrams_1D}
}

\subfloat[Two-qubit gate set of Toffoli circuit]{
\includegraphics[scale=1]{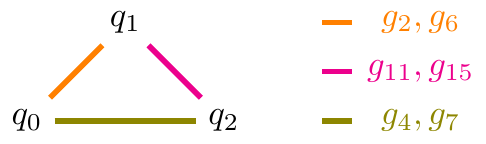}
\label{fig:2qbg-set}
}
\end{minipage}

\caption{Toffoli circuit (Single-qubit gates are colored gray. Identical two-qubit gates applied at different times have the same color, e.g., $g_2$ and $g_6$ are orange because they are both $CX(q_1, q_2)$ but at different times.)}
\label{fig:QCDiagrams}
\end{figure}

It is important to note that the qubits in a quantum circuit are {\it logical qubits} denoted as $q_1,..., q_N$. 
We do not use ``logical'' to indicate error correction in this paper. 
Additionally, we denote {\it qubit count} by $N$, {\it gate count} by $M$, {\it single-qubit gate count} by $M_1$, and {\it two-qubit gate count} by $M_2$. 
For instance, in Fig.~\ref{fig:QCDiagrams}, $N=3$, $M=15$, $M_1 = 9$, $M_2 = 6$.

We use the notation of gate like a set of qubits. 
We say the cardinality $|g|=1$ if $g$ is a single-qubit gate and $|g|=2$ if $g$ is a two-qubit gate. 
$g\cup g'$ denotes the set of qubits involved in $g$ or $g'$.
$g\cap g'$ denotes the set of qubits involved in $g$ and $g'$. 
For instance, in Fig.~\ref{fig:QCDiagrams}, $q_1, q_2\in g_2$, $q_0\in g_{13}$, $|g_4|=|g_6|=2$, $|g_7|=|g_9|=1$, $g_1\cap g_2=\{q_2\}$, and $g_8\cup g_9=\{q_0, q_1, q_2\}$.

A 1D circuit diagram can also represent the circuit, e.g., Fig.~\ref{fig:QCDiagrams_1D}. 
In such a diagram, each wire stands for a logical qubit. 
We draw the control qubit of $CX$ gate as $\bullet$ and the target qubit as $\oplus$.
The 1D diagram provides some implicit timing information, i.e., the gates on the same horizontal wire are executed from left to right. 
However, vertical alignment has no indication on timing, e.g., $g_8$ and $g_9$ can be executed simultaneously, but we separate them by some horizontal distance purely to avoid overlapping in the diagram.

\subsection{Quantum Computing Device Representation}

We represent the layout of a QC device with a {\it device graph} $G=(P, E)$ where each node $p$ stands for a {\it physical qubit} and each edge $e$ stands for a connection that enables two-qubit entangling gates, i.e., we can only perform such gates on two physical qubits that are connected. 
This graph is also named as {\it coupling graph} or {\it qubit connectivity}. 
Device graphs used in this paper are shown in Fig.~\ref{fig:layouts}.

In the coupling graphs, we assume the edges are bi-directional.
We made this decision based on the following reasons.
First, if the two-qubit gate in the gate library is symmetric, like $CZ$ gate used by Google and Rigetti as shown in Table~\ref{tab:gates}, then it does not make sense to have directed connections.
Second, even for asymmetric gates like $CX$, most of the IBM quantum computers\footnote{\url{https://quantum-computing.ibm.com/}} now support both directions for every edge.
Third, as we shall see in Sec.~\ref{sec:approach}, our construction of QUEKO benchmarks works for directed graphs as well.
So, the assumption of bi-directionality does not change the conclusion of this study.

\begin{figure}[tb]
\centering
\begin{minipage}[]{0.3\linewidth}
\centering
\subfloat[IBM's Ourense device graph]{
\begin{tikzpicture}[x=0.8cm,y=0.8cm]
\node[draw=cyan, circle, ultra thick] (q4) at (0, 1) {$p_4$};
\node[draw=cyan, circle, ultra thick] (q2) at (2, 1) {$p_2$};
\node[draw=cyan, circle, ultra thick] (q3) at (0, 0) {$p_3$};
\node[draw=cyan, circle, ultra thick] (q1) at (2, 0) {$p_1$};
\node[draw=cyan, circle, ultra thick] (q0) at (4, 0) {$p_0$};

\draw[cyan,-,ultra thick] (q0)--(q1);
\draw[cyan,-,ultra thick] (q1)--(q2);
\draw[cyan,-,ultra thick] (q1)--(q3);
\draw[cyan,-,ultra thick] (q3)--(q4);

\end{tikzpicture}
\label{fig:Ourense_layout}
}
\vspace{0.1cm}
\subfloat[Rigetti's Aspen-4 device graph]{
\includegraphics[width=\linewidth]{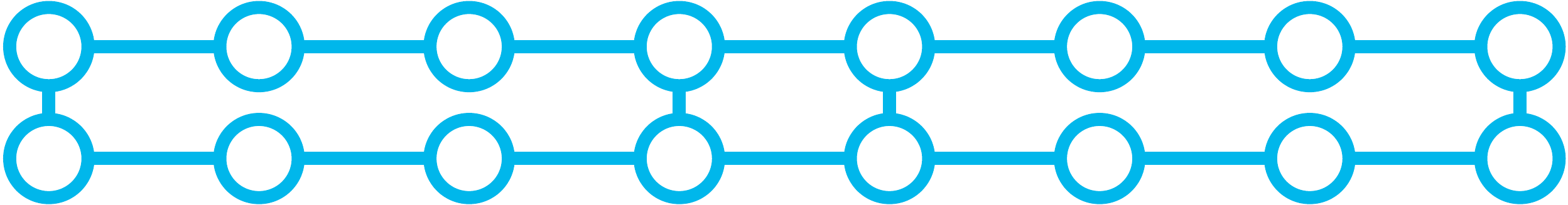}
}
\end{minipage}
\hspace{0.1cm}
\begin{minipage}[]{0.3\linewidth}
\subfloat[IBM's Tokyo device graph]{
    \includegraphics[width=\linewidth]{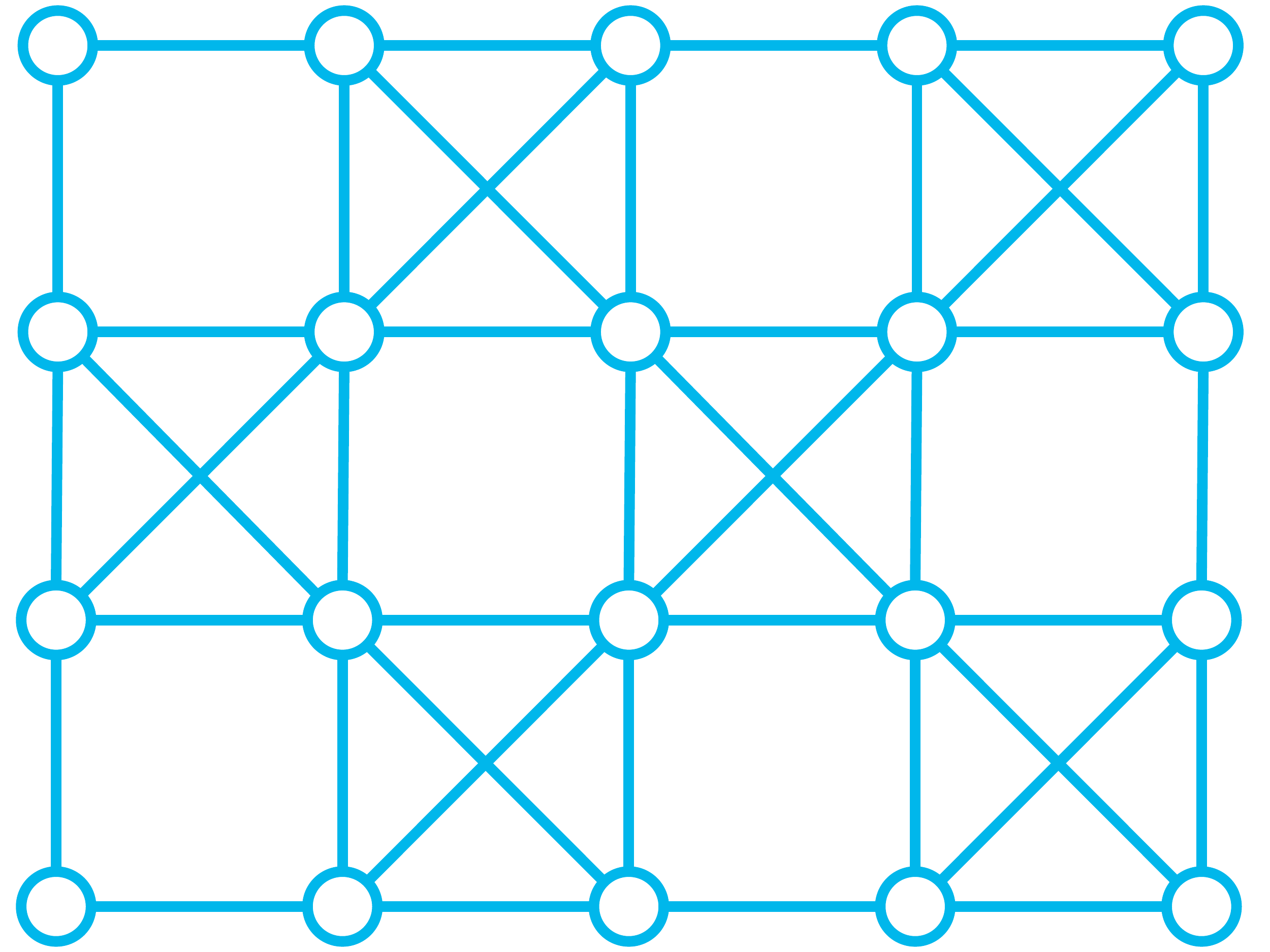}
    \label{fig:Tokyo_layout}
}
\end{minipage}
\vfill
\subfloat[IBM's Rochester device graph]{
    \includegraphics[width=0.3\linewidth]{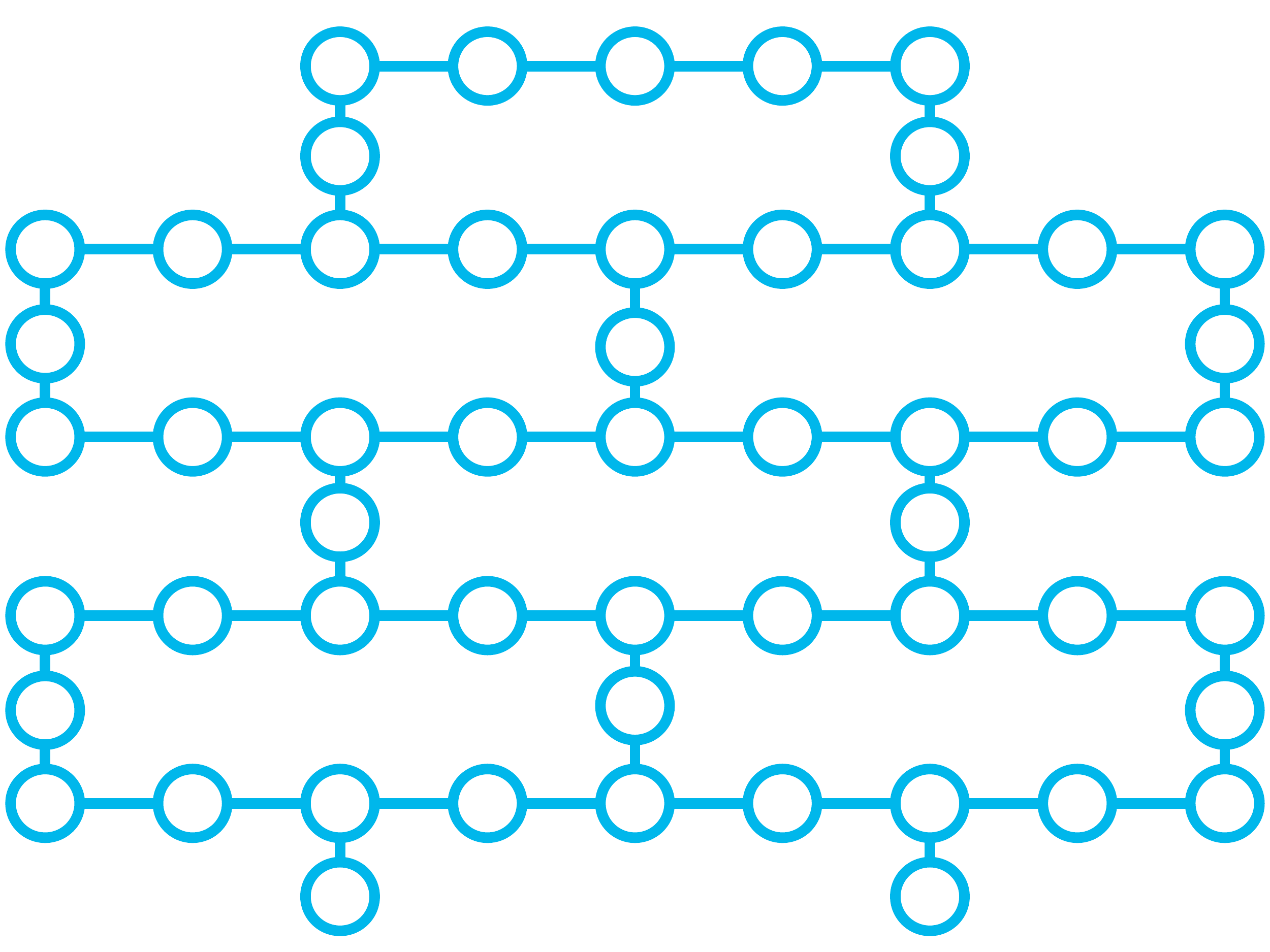}
}
\hspace{0.1cm}
\subfloat[Google's Sycamore device graph]{
    \includegraphics[width=0.3\linewidth]{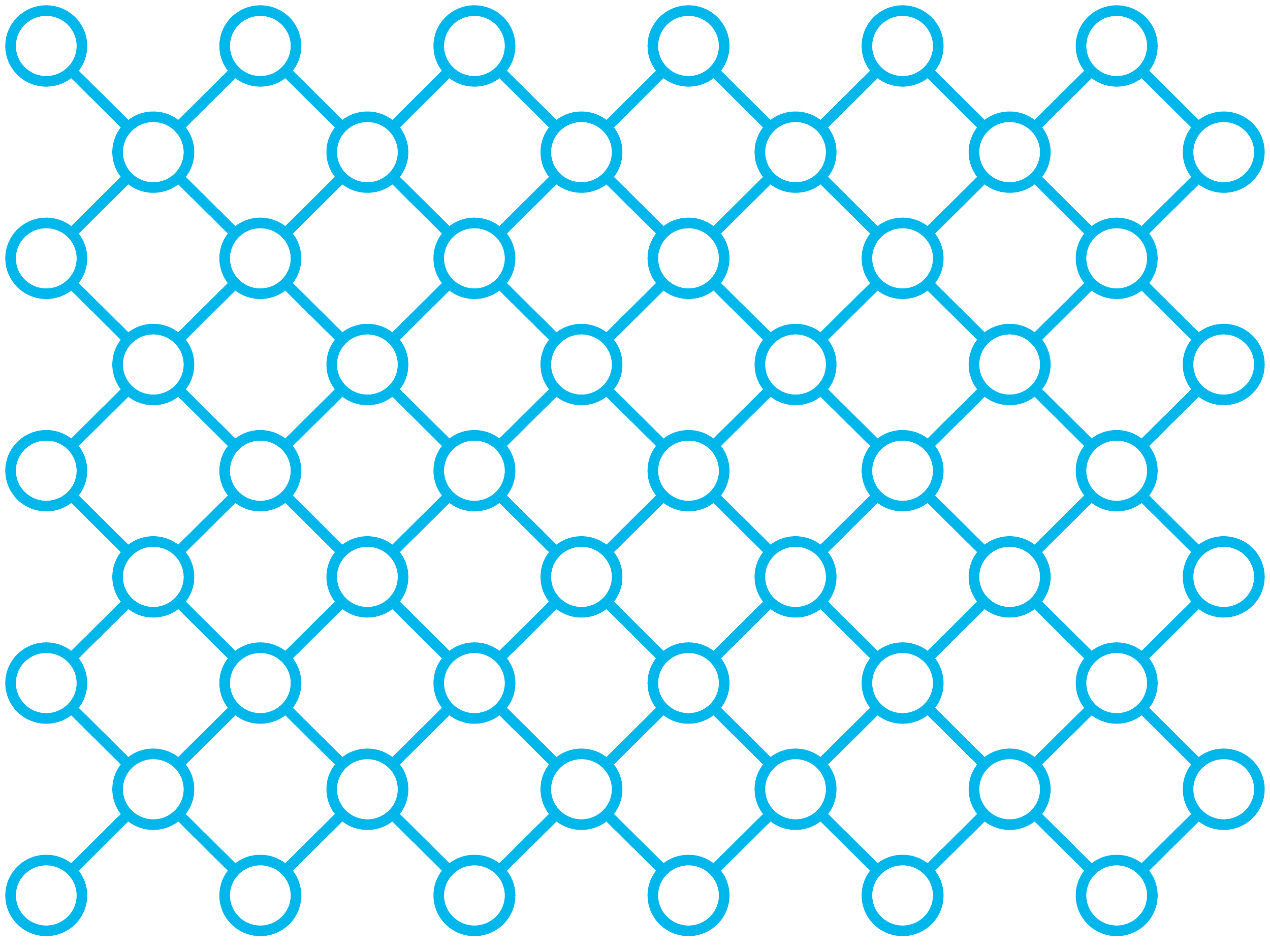}
}

\caption{\label{fig:layouts} QC device examples}
\end{figure}
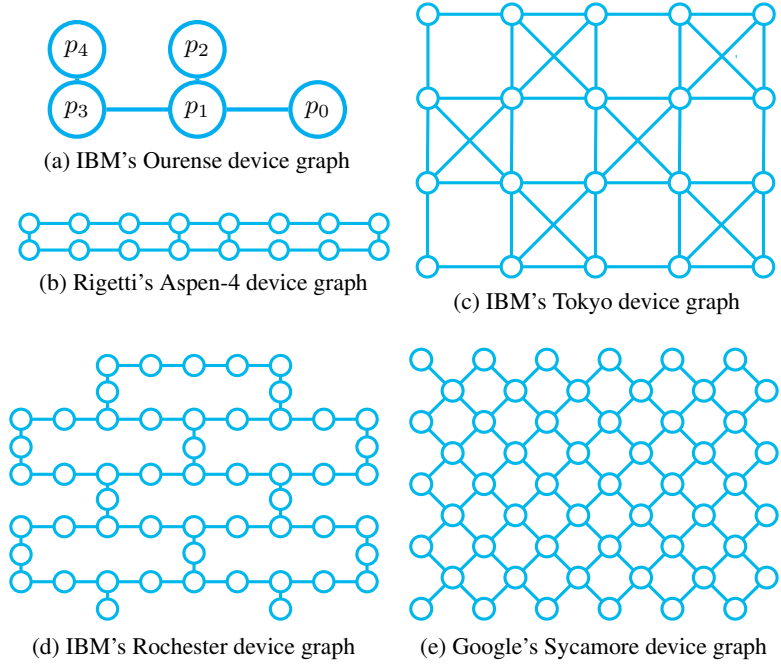

\section{Problem Formulation}\label{sec:problem}
Layout synthesis is divided into two sub-tasks: {\it initial placement} that produces an initial mapping from logical qubits to physical qubits, and {\it gate scheduling} that decides when and where to execute each input gate and insert SWAP gates to make two-qubit gates legal on the device graph.
However, a SWAP gate may increase the depth, and definitely increase the gate count and the error rate.
In the whole QC compilation/mapping workflow, there are usually circuit optimization stages before or after layout synthesis.
During the optimization stages, gate reduction and commutation are performed, e.g.,~\cite{nam-2018-automated}.
We assume that these optimizations are already applied prior to layout synthesis, so that every input gate should be executed and the relative order of input gates should not change in the layout synthesis process.

\begin{figure}[htb]
    \centering
    \subfloat[Initial placement for Toffoli circuit]{
        \includegraphics[width=0.6\linewidth]{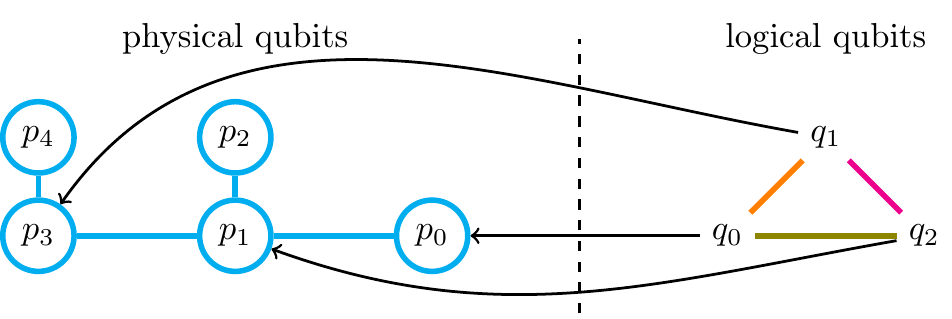}
        \label{fig:placement}
    }
    
    \subfloat[Scheduled Toffoli circuit]{
        \includegraphics[width=0.9\linewidth]{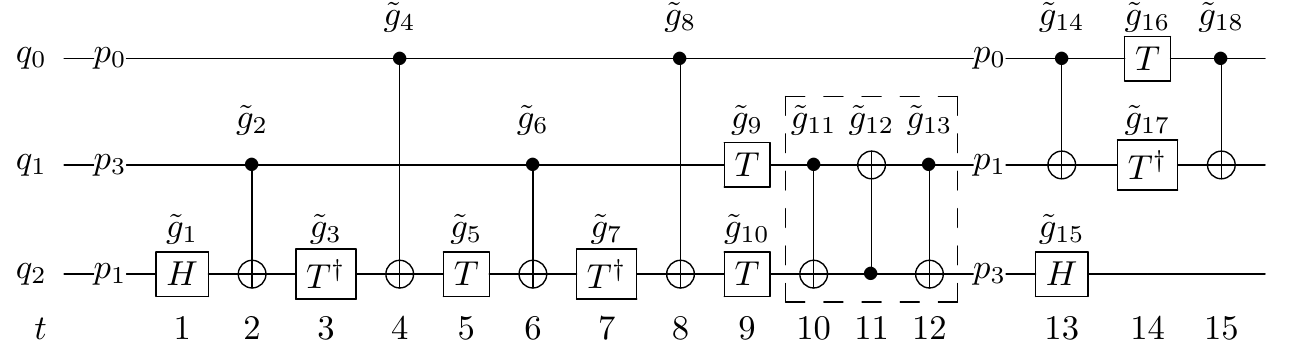}
        \label{fig:scheduling}
    }
    \caption{Quantum computing layout synthesis for Toffoli circuit on IBM's Ourense device}
    \label{fig:LayoutSynthesis}
\end{figure}

\subsection{Initial Placement}

During initial placement, we need to find a mapping from logical qubits in the quantum circuit to physical qubits on the device $\mu_0:Q\to P$ that benefits subsequent gate scheduling. 
If the two-qubit gate set, consisting of all the two-qubit gates, can be embedded in the device graph during initial placement (e.g., Fig.~\ref{fig:2qbg-set} can be embedded in Fig.~\ref{fig:Tokyo_layout}), gate scheduling does not necessarily need to insert any gates. 
However, the case is not so ideal in general. 
For example, if we want to map the Toffoli circuit as shown in  Fig.~\ref{fig:QCDiagrams_1D} onto device Ourense as shown in Fig.~\ref{fig:Ourense_layout}, there must be some additional gates, since the device graph does not contain any triangles, but the two-qubit gates in Fig.~\ref{fig:2qbg-set} forms a triangle. 
A valid initial placement in this case is given in Fig.~\ref{fig:placement} where $\mu_0(q_0)=p_0$, $\mu_0(q_1)=p_3$, and $\mu_0(q_2)=p_1$.

\subsection{Gate Scheduling}

Given a quantum circuit $g_1...g_M$, e.g., Fig.~\ref{fig:QCDiagrams_1D}, gate scheduling produces the {\it spacetime coordinates} $(t_j,x_j)$ for each gate. 
The coordinates specify when and where the gates are applied. 
We say that a gate is scheduled to {\it cycle} $t$ if its time coordinate is $t$.
For a single-qubit gate, the space coordinate is a physical qubit, i.e., $x_j\in P$; for a two-qubit gate, it is an edge in the device graph, i.e., $x_j\in E$. 
SWAP gates may need to be inserted during gate scheduling to ensure that all two-qubit gates are executable. 
The input gate plus the inserted SWAP gates constitute the scheduled gate list $\tilde g_1,..., \tilde g_{\tilde M}$.
Since only SWAP gates are inserted and all the input gates are contained in the scheduled gate list, the functionality of input circuit remains unchanged after the layout synthesis process. 
Additionally, gate scheduling must respect {\it dependencies} in the quantum circuit. 
If gate $g$ acts on qubit $q$, then $g$ can only be executed after all prior gates, which act on qubit $q$, are executed.

A valid but not necessarily optimal gate scheduling example is given in Fig.~\ref{fig:scheduling}.
The time coordinates for all the gates are displayed at the bottom, e.g., $t_1=1$ and $t_{15}=13$.
The space coordinates can be inferred from the mapping, e.g. $x_1=p_1$, $x_2=(p_3,p_1)$, $x_{14}=(p_0,p_1)$, and $x_{15}=p_3$.
There is an injective map from the original gates to the scheduled gates: $f(i)=i$ for $i=1$ to $10$ and $f(i) = i + 3$ for $i=11$ to $15$ such that $g_i=\tilde g_{f(i)}$ for $i=1$ to $15$.
The three $CX$ gates in the dashed box $\tilde g_{11}$, $\tilde g_{12}$, and $\tilde g_{13}$ constitute a SWAP gate.
The adjusted mapping is shown after the them.
The SWAP gate is inserted so that $\tilde g_{14}$ and $\tilde g_{18}$ are on connected qubits $p_0$ and $p_1$, thus executable.

\subsection{Formal Definition of (Depth-Optimal) Layout Synthesis Problem in Quantum Computing}

{\bf Input} A device graph $G=(P, E)$ and a list of quantum gates $g_1...g_M$ acting on logical qubit set $Q$. All the input gates are in the gate set of the device, e.g., a set from Table \ref{tab:gates}. Logical qubits are less or equal than physical qubits, i.e., $|Q|\le |P|$.

\noindent {\bf Output} An initial mapping $\mu_0:Q\to P$, and a scheduled quantum circuit consists of a new list of gates $\tilde g_1...\tilde g_{\tilde M}$, including SWAP gates, where each gate has a spacetime coordinate $(t_j, x_j)$. We use tilde to denote that a gate is scheduled from here on.

\noindent {\bf Constraints} 
\begin{itemize}
    \item {\bf Feasible two-qubit gates:} all the two-qubit gates in the scheduled circuit must be on two qubits connected in the device graph.
    Formally, for $j=1$ to $\tilde M$, if $|\tilde g_j|=2$, then $x_j\in E$.
    \item {\bf Executing all gates:} all input gates should be executed. 
    Formally, there is an injective map $f:\{1,...,M\}\to\{1,...,\tilde M\}$ such that $g_i = \tilde g_{f(i)}$ for $i=1$ to $M$.
    \item {\bf Respecting dependencies:} for $i,i'=1$ to $M$, if $i<i'$ and $g_i\cap g_{i'}\neq\emptyset$ then $t_{f(i)}<t_{f(i')}$.
\end{itemize}

\noindent {\bf Objective} Minimizing circuit {\it depth} $T$, which is the maximal time coordinate of all the scheduled gates, i.e., $T \equiv \max_{j=1,...,\tilde M} t_j$.

In this paper, we use depth as the default objective but other objectives can be used as well, e.g., the number of additional gates $\tilde M - M$, or the fidelity of the scheduled circuit. The output and the constraints of the problem are independent of the objective. However, with other objectives like fidelity, more input information may be required.

\section{Related Work}\label{sec:related}

In the most general sense, the task of QC layout synthesis is generating a quantum circuit that satisfies QC device constraints and fulfills the functionality of the input circuit. 
Related works on this problem include~\cite{whitney-2007-automated, maslov-2008-quantum, hirata-2011-efficient, shafaei-2013-optimization, shafaei-2014-qubit, kole-2016-heuristic, kole-2018-new, wille-2014-optimal, siraichi-2018-qubit, zulehner-2018-efficient, childs-2019-circuit, itoko-2019-quantum, wille-2019-mapping, siraichi-2019-qubit, cowtan-2019-qubit, bhattacharjee-2019-muqut, lye-2015-determining, venturelli-2017-temporal, venturelli-2018-compiling, booth-2018-temporal, zulehner-2019-compiling, kole-2019-improved, tannu-2019-all, li-2019-tackling, murali-2019-noise, murali-2019-full, ash-2019-qure}. 
These works may have some variations on the problem in mind. 
\cite{whitney-2007-automated, hirata-2011-efficient, shafaei-2013-optimization, shafaei-2014-qubit, lye-2015-determining, kole-2016-heuristic, kole-2018-new, murali-2019-noise} focus on multidimensional array device graphs (linear array for 1D, grid for 2D, and so on). 
\cite{zulehner-2019-compiling} focuses specifically on SU(4) circuits and includes post-synthesis optimization. 
\cite{ash-2019-qure} focuses on adjusting the mapping after synthesis to improve fidelity.

As pointed out in Sec.~\ref{sec:gates}, it is nontrivial to judge if two quantum gates commute.
If two gates do not commute, changing their relative order would result in a functionally different circuit.
Therefore, in general, performing commutations are seen to be part of logic synthesis for quantum computing, which happens before layout synthesis, the main topic of this paper.
Thus, to better isolate layout synthesis, we stated in Sec.~\ref{sec:problem} that, as an assumption, all commutation-based optimization have been performed beforehand.
This assumption is also respected in all of the evaluations afterwards and all the other related works listed above except~\cite{itoko-2019-quantum, venturelli-2017-temporal, venturelli-2018-compiling, booth-2018-temporal}.
\cite{itoko-2019-quantum} considers a few basic commutation rules, but the solutions it compares to do not use dependencies to encode the relative orders between the gates.
Thus, it is unclear whether the gain of~\cite{itoko-2019-quantum} is because of commutations, or encoding with dependencies.
\cite{venturelli-2017-temporal, venturelli-2018-compiling, booth-2018-temporal} consider the layout synthesis of QAOA circuits. 
All the two-qubit gates used in QAOA circuits~\cite{farhi-2014-qaoa} are $ZZ$-phase gates that commute with each other, which means the relative order of two $ZZ$-phase gates can be changed even if they act on a same qubit.
However, these gates rarely appear in other existing quantum circuits, so we decide to treat QAOA circuits as a special case instead of jeopardizing our efforts in isolating layout synthesis.

Existing works on quantum layout synthesis can be characterized by based on optimization metrics used. 
The metric can be the additional ``cost'', which is usually proportional to the number of additional gates~\cite{maslov-2008-quantum, hirata-2011-efficient, shafaei-2013-optimization, shafaei-2014-qubit, lye-2015-determining, kole-2016-heuristic, kole-2018-new, siraichi-2018-qubit, childs-2019-circuit, cowtan-2019-qubit,  itoko-2019-quantum, zulehner-2018-efficient, zulehner-2019-compiling, siraichi-2019-qubit, wille-2014-optimal, wille-2019-mapping}; or circuit depth like this paper, since the qubits can only function well within the decoherence time~\cite{whitney-2007-automated,  venturelli-2017-temporal, booth-2018-temporal, venturelli-2018-compiling, bhattacharjee-2019-muqut }; or circuit fidelity, since nowadays a common practice is executing a circuit multiple times and analysing the statistics of the results~\cite{tannu-2019-all, ash-2019-qure, murali-2019-noise, murali-2019-full}; or a mix of the above~\cite{li-2019-tackling, kole-2019-improved}.

Detailed discussions of the complexity of QC layout synthesis can be found in Sec.~\ref{sec:complexity}.
These discussions indicate that large scale instances cannot be solved both exactly and efficiently. 
From the perspective of solution techniques, the current works can be divided into two categories. 
The first group focuses on deriving the exact solution for moderate-sized instances with the help of solvers~\cite{bhattacharjee-2019-muqut, booth-2018-temporal, venturelli-2017-temporal, venturelli-2018-compiling, lye-2015-determining, wille-2019-mapping, wille-2014-optimal, shafaei-2014-qubit, murali-2019-noise, murali-2019-full}. \cite{lye-2015-determining, wille-2014-optimal} use a PBO (pseudo Boolean optimizer) to decide a set of SWAP gates that legalizes all the two-qubit gates, but they do not explicitly schedule the gates. The same goes for~\cite{shafaei-2014-qubit, wille-2019-mapping}, which use a MIP (mixed integer programming) solver and a SMT (satisfiability modulo theories) solver correspondingly. \cite{booth-2018-temporal, venturelli-2017-temporal, venturelli-2018-compiling} use a temporal planner to schedule specifically QAOA circuits. The closest previous works concerning this paper are~\cite{bhattacharjee-2019-muqut, murali-2019-noise, murali-2019-full}. However, \cite{murali-2019-noise, murali-2019-full} use a SMT solver to maximize fidelity. \cite{bhattacharjee-2019-muqut} splits circuits into ``levels'' and inserts gates to transform the mapping between the levels. This model of quantum circuit may not yield an optimal solution. Under this imperfect ``levels'' model, \cite{bhattacharjee-2019-muqut} aims to derive a depth-optimal solution with integer linear programming.

The second group of related works use heuristic search techniques~\cite{zulehner-2018-efficient, maslov-2008-quantum, whitney-2007-automated, kole-2019-improved, li-2019-tackling, hirata-2011-efficient, shafaei-2013-optimization, kole-2016-heuristic, kole-2018-new, siraichi-2018-qubit, childs-2019-circuit, cowtan-2019-qubit, itoko-2019-quantum, zulehner-2019-compiling, tannu-2019-all, siraichi-2019-qubit}. 
We only discuss the works targeting general device layouts below~\cite{zulehner-2018-efficient, maslov-2008-quantum, kole-2019-improved, li-2019-tackling, siraichi-2018-qubit, siraichi-2019-qubit,  childs-2019-circuit, cowtan-2019-qubit, itoko-2019-quantum, zulehner-2019-compiling, tannu-2019-all}.
The general approach is splitting the circuit into small sub-circuits for which the layout synthesis can be done efficiently, and then searching for the mapping transformation between these sub-circuits. 
A sub-circuit can be a ``level'' or ``layer'' mentioned in the last paragraph~\cite{zulehner-2018-efficient, siraichi-2019-qubit, childs-2019-circuit, li-2019-tackling, cowtan-2019-qubit, kole-2019-improved, tannu-2019-all}, a set of several levels~\cite{maslov-2008-quantum}, several levels but for a few specific qubits~\cite{zulehner-2019-compiling}, or individual gates~\cite{siraichi-2018-qubit, itoko-2019-quantum}.
In order to find the mapping transformation,~\cite{siraichi-2018-qubit} inserts SWAP gates to move the two qubits that required by the next two-qubit gate in the shortest path; \cite{cowtan-2019-qubit, itoko-2019-quantum, kole-2019-improved} also consider distances between qubits of further two-qubit gates; \cite{tannu-2019-all} additionally considers fidelity in the qubit movements; \cite{zulehner-2018-efficient, zulehner-2019-compiling} use the sum-of-distances plus the number of SWAP gates as the cost function in A* search; \cite{li-2019-tackling} uses bidirectional search; \cite{childs-2019-circuit} uses a 4-approximation algorithm; \cite{siraichi-2019-qubit} exploits some existing approximate solution of token swapping; \cite{maslov-2008-quantum} recursively considers SWAP gates as cuts in the device graph. 

The complexity also brings difficulty to the evaluation of these solutions. 
Currently, the benchmarks usually are quantum circuit libraries of some realistic functions, e.g., \cite{nam-2018-automated} and RevLib~\cite{wille-2008-revlib}, or certain random circuits, which are thought to be the worst-case scenario, e.g., SU(4) circuits~\cite{zulehner-2019-compiling}. 
So far, researchers can only compare against each other, but do not know how far they are from the optimum. 
This paper aims to fill in this gap.

The QC layout synthesis problem is still quite new to compiler and design automation communities, so the name of the problem varies. 
It can be placement~\cite{maslov-2008-quantum, shafaei-2014-qubit}, routing~\cite{cowtan-2019-qubit}, compiling quantum circuits~\cite{itoko-2019-quantum, zulehner-2019-compiling, murali-2019-noise, venturelli-2017-temporal, venturelli-2018-compiling, booth-2018-temporal, murali-2019-full}, quantum circuit transformation~\cite{childs-2019-circuit}, mapping circuits to QC architectures~\cite{zulehner-2018-efficient, wille-2014-optimal, lye-2015-determining, wille-2019-mapping, kole-2019-improved, bhattacharjee-2019-muqut, li-2019-tackling}, conversion~\cite{hirata-2011-efficient} or optimization~\cite{shafaei-2013-optimization} of circuits in QC architecture, realization of quantum circuits~\cite{kole-2016-heuristic, kole-2018-new}, or qubit allocation~\cite{siraichi-2018-qubit, tannu-2019-all, ash-2019-qure, siraichi-2019-qubit}.

\section{QUEKO Benchmarks}\label{sec:approach}
This work is inspired by PEKO~\cite{chang-2004-optimality}, placement examples with known optimal. 
Placement is a crucial step in classical integrated circuit design, where modules are placed on a chip with objective of minimizing total wirelength. 
Although this problem is NP-hard, the PEKO algorithm is able to generate benchmarks with know optimal solutions.

Similarly, for a generic input quantum circuit and a generic device graph, finding the scheduled circuit with optimal depth is NP-complete, which will be proved in Sec.~\ref{sec:complexity}. 
However, it is feasible to construct some benchmarks with known optimal solution. 
Given a target device graph $G$ and a target depth $T$, we can construct an depth-optimal circuit. 
Then, by re-labelling the qubits, we derive a QUEKO benchmark.

Additionally, QUEKO can be customized for a given feature: gate density vector $(d_1, d_2)$. 
The two components intuitively stand for the densities of single-qubit gates and two-qubit gates in the whole circuit. 
Suppose a circuit has $n$ logical qubits, $M_1$ single-qubit gates, $M_2$ two-qubit gates, and a longest dependency chain of $l$, then $d_1=M_1/(n\cdot l)$ and $d_2=2M_2/(n\cdot l)$.
For example, in Fig.~\ref{fig:QCDiagrams_1D}, $n=3$, $l=11$, $M_1=9$, and $M_2=6$, so $d_1\approx 0.27$ and $d_2\approx 0.36$.
Likewise, we can extract $(d_1, d_2)$ from other existing circuits with known functionalities, so that the QUEKO benchmarks imitate some real-world circuits and, at the same time, have known optimal depths.

The construction of QUEKO, as shown in Algorithm \ref{alg:construction}, starts with checking the validity of input data by calculating the number of single-qubit and two-qubit gates $M_1$ and $M_2$. 
If $M_1+M_2<T$, then there would be too few gates to generate a circuit with depth $T$; if $M_1+2M_2>N\cdot T$, then there would be too many gates for the given depth and device graph. 
We define the matching bound $u$ of a graph $G$ to be the minimal size of maximal matchings of $G$. 
This means we can find at least $u$ edges in $G$ that pair-wisely share no vertices.
If $M_2>u\cdot T$, then there could be too many two-qubit gates for the given depth and device graph.
In short, if $M_1+M_2<T$, $M_1+2M_2>N\cdot T$, or $M_2>u\cdot T$, we return an error to reject the input data.
Otherwise, we proceed to three phases: backbone construction, sprinkling, and scrambling. 

\begin{algorithm}[htbp]
\begin{algorithmic}[1]
\REQUIRE a device graph $G = (P, E)$ with $|P|=N$ and its matching bound $u$, a depth target $T$, and a gate density vector $(d_1, d_2)$
\ENSURE QUEKO benchmark $g_1g_2...g_{M_1+M_2}$, where $M_1$ and $M_2$ are the numbers of single-/two-qubit gates

\STATE $M_1\gets\lceil d_1\cdot N\cdot T\rceil$, $M_2\gets\lceil d_2\cdot N\cdot T/2\rceil$
\IF{$M_1+M_2< T$ \OR $M_1+2M_2>N\cdot T$ \OR $M_2>u\cdot T$}
    \RETURN error: input data not admissible
\ENDIF
\STATE $m_1 \gets 0$, $m_2\gets0\quad$ \COMMENT{how many single-qubit gates and two-qubit gates we have used}

\COMMENT {Backbone construction phase}
\FOR {$i=1$ to $T$}
    \STATE $j\gets\text{rand}(\{1, 2\})\quad$ \COMMENT{randomly decide single-qubit or two-qubit gate}
    \IF{$j=2$ \AND $m_2<M_2$}
        \STATE $x_i \gets \text{rand}(E)$
        \WHILE{$i>1$ \AND $x_i\cap x_{i-1}=\emptyset$}
            \STATE $x_i \gets \text{rand}(E)$ 
        \ENDWHILE
        \STATE $t_i\gets i$, $m_2\gets m_2+1$
    \ELSE
        \STATE $x_i \gets \text{rand}(P)$
        \WHILE{$i>1$ \AND $x_i\cap x_{i-1}=\emptyset$}
            \STATE $x_i \gets \text{rand}(P)$
        \ENDWHILE
        \STATE $t_i\gets i$, $m_1\gets m_1+1$
    \ENDIF
\ENDFOR

\COMMENT {Sprinkling phase}
\FOR {$i=T+1$ to $M_1+M_2$}
    \STATE $j\gets\text{rand}(\{1, 2\})$
    \IF{$j=2$ \AND $m_2<M_2$} 
        \STATE $(t_i, x_i) \gets \text{rand}(\{1,2,...,T\}\times E)$
        \WHILE{$\exists l\in\{1,...,i\}$ such that $t_i=t_l$ \AND $x_i \cap x_l \neq\emptyset$}
            \STATE $(t_i, x_i) \gets \text{rand}(\{1,2,...,T\}\times E)$
        \ENDWHILE
        \STATE $m_2\gets m_2+1$
    \ELSE
        \STATE $(t_i, x_i) \gets \text{rand}(\{1,2,...,T\}\times P)$
        \WHILE{$\exists l\in\{1,...,i\}$ such that $t_i=t_l$ \AND $x_i \cap x_l \neq\emptyset$}
            \STATE $(t_i, x_i) \gets \text{rand}(\{1,2,...,T\}\times P)$
        \ENDWHILE
        \STATE $m_1\gets m_1+1$
    \ENDIF
\ENDFOR

\COMMENT {Scrambling phase}
\STATE $\tau\gets$ a random mapping from $P$ to $Q$
\FOR{$i=1$ to $M_1+M_2$}
    \IF{$x_i=(p, p')\in E$}
        \STATE $g_i\gets\text{two-qubit gate}(\tau(x_i.p), \tau(x_i.p'))$
    \ELSE
        \STATE{$g_i\gets\text{single-qubit gate}(\tau(x_i))$}
    \ENDIF
\ENDFOR

\COMMENT{Output}
\STATE {\bf sort} $g_i$ according to $t_i$, $i=1$ to $M_1+M_2$
\RETURN $g_1g_2...g_{M_1 + M_2}$
\end{algorithmic}
\caption{QUEKO construction}
\label{alg:construction}
\end{algorithm}

\begin{figure}[htb]
\centering
\subfloat[Device graph]{
\includegraphics[width=0.3\linewidth, rotate=90]{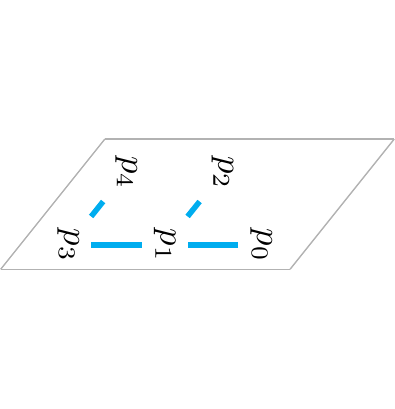}
\label{fig:Ourense_rotated}
}
\hspace{1em}
\subfloat[Backbone construction phase]{
\includegraphics[width=0.3\linewidth, rotate=90]{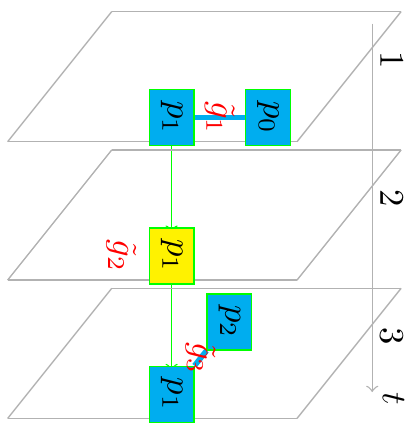}
\label{fig:Construction_Grow}
}

\subfloat[Sprinkling phase]{
\includegraphics[width=0.3\linewidth, rotate=90]{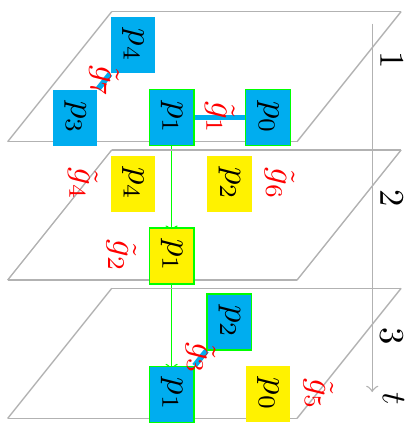}
\label{fig:Construction_Sprinkle}
}
\hspace{1em}
\subfloat[Scrambling phase]{
\includegraphics[width=0.3\linewidth, rotate=90]{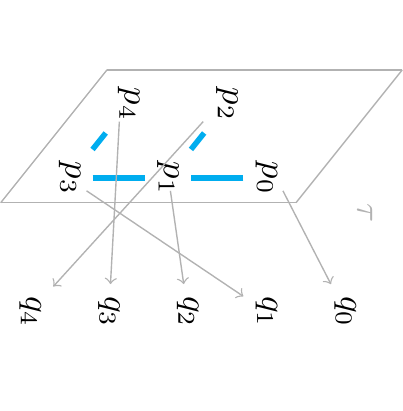}
\label{fig:Construction_Shuffle}
}

\begin{minipage}{0.3\linewidth}
\subfloat[Output circuit 1D diagram]{
\includegraphics[width=\linewidth]{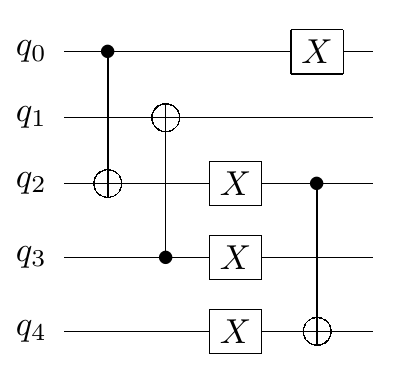}
\label{fig:Construction_Output}
}
\end{minipage}
\hspace{1em}
\begin{minipage}[b]{0.3\linewidth}
\subfloat[Output QASM code (Inside the parentheses are the corresponding gates before the scrambling phase.)]{
\begin{lstlisting}
OPENQASM 2.0;
include "qelib1.inc";
qreg q[5];
cx q[0], q[2]; //(*$g_1(\tilde g_{1})$*)
cx q[3], q[1]; //(*$g_2(\tilde g_{7})$*)
x q[2]; //(*$g_3(\tilde g_{2})$*)
x q[3]; //(*$g_4(\tilde g_{4})$*)
x q[4]; //(*$g_5(\tilde g_{6})$*)
cx q[2], q[4]; //(*$g_6(\tilde g_{3})$*)
x q[0]; //(*$g_7(\tilde g_{5})$*)
\end{lstlisting}
\label{fig:Construction_qasm}
}
\end{minipage}

\caption{QUEKO construction visualization}
\label{fig:Construction}
\end{figure}

\subsection{Backbone Construction Phase}

This phase ``grows'' a sequence of $T$ gates, each depending on the previous one, constituting a dependency chain of length $T$. 
This chain serves as the ``backbone'' of the circuit.
For example, we start from the device graph as Fig.~\ref{fig:Ourense_rotated} (which is just Fig.~\ref{fig:Ourense_layout} rotated), and pick three executable gates $\tilde g_1$, $\tilde g_2$, and $\tilde g_3$ whose spacetime coordinates are $(1, (p_0, p_1))$, $(2, p_1)$, $(3, (p_1, p_2))$. 
They constitute a dependency chain of length $T=3$, since all of them act on $p_1$.
This is shown in Fig.~\ref{fig:Construction_Grow}, where gates at different cycles are put on different ``slices'' from left to right. 
The ``backbone'' is colored green.
Note that if the given graph is directed, then we make sure that the direction of the two-qubit gate we choose is consistent with the direction of the corresponding edge.
Thus, our construction also works for directed device graphs.

To be more rigorous, we first choose a random node or edge of $G$ as $x_1$.
In every iteration afterwards, we randomly choose $x_k$ that overlaps with $x_{k-1}$.
Thus, $x_k\cap x_{k-1}\neq\emptyset$, which enforces $t_k>t_{k-1}= k-1$ by dependency constraint.
On the other hand, since $\tilde g_k$ is executable, it can at most take a single cycle, i.e., the optimal $t_k=t_{k-1} + 1=k$.
Gate sequence $\tilde g_1, \tilde g_2,..., \tilde g_T$ constitutes a dependency chain of length $T$.
Because of this ``backbone'', the final depth of the scheduled circuit cannot be lower than $T$.
Note that we do not need to use any SWAP gates for backbone construction.

\subsection{Sprinkling Phase} 

The backbone construction phase uses $T$ gates in total, we then randomly ``sprinkle'' the rest $M_1+M_2-T$ gates, e.g., $\tilde g_4,\tilde g_5, \tilde g_6, \tilde g_7$ shown in Fig.~\ref{fig:Construction_Sprinkle}.
We randomly select spacetime coordinates $(t_i, x_i)$, ($1\le t_i \le T$) that does not overlap with any existing gates with time coordinate $t_i$.
Since all the time coordinates of gates sprinkled are less or equal to $T$, the backbone is not lengthened in this phase.
After sprinkling, a circuit with gates $\tilde g_1...\tilde g_{M_1+M_2}$ is created.
Its gates are all executable; its depth is $T$; its gate density vector approximates $(d_1, d_2)$.
(There could be minor rounding errors in the ceiling function.)

It is worthy of noting that though only one longest dependency chain is explicitly generated in the backbone construction phase, the sprinkling phase may implicitly generate more.
For example, $\tilde g_4$ depends on $\tilde g_7$; if we ``sprinkles'' a gate on $p_4$ at cycle $3$, then another dependency chain of length $3$ would exist in the output circuit.
The higher the gate densities, the more likely that these implicit longest dependency chains are generated.

\subsection{Scrambling Phase}

As shown in Fig.~\ref{fig:Construction_Shuffle}, we generate a random mapping $\tau$ from physical qubits to logical qubits and apply $\tau$ to the space coordinates of $\tilde g_1\tilde g_2...\tilde g_{M_1+M_2}$.
For instance, $x_1=(p_0, p_1)$, so the resulting gate $g_1$ is a two-qubit gate on logical qubits $\tau(p_0)=q_0$ and $\tau(p_1)=q_2$; $x_7=(p_3, p_4)$, so $g_7$ is a two-qubit gate on $\tau(p_3)=q_1$ and $\tau(p_4)=q_3$; $g_6$ is a single-qubit gate on $\tau(p_2)=q_4$...
The specific types of single-qubit gates and two-qubit gates are not important, since QUEKO is only for layout synthesis, not for circuit optimization.
We use $X$ as the single-qubit gate and $CX$ as the two-qubit gate.

\subsection{Output}
Sort the gates $g_1g_2...g_{M_1+M_2}$ according to the time coordinates to transfer the timing information originally in these time coordinates to the relative order inside the output gate list.
The result is a QUEKO benchmark, as shown in Fig.~\ref{fig:Construction_Output} and Fig.~\ref{fig:Construction_qasm}.

As we have proven, the depth of the output circuit is at least $T$ because of the backbone. 
A QC layout synthesis tool can meet the optimum by finding the initial mapping that is the inverse of the scrambling mapping $\tau$. 
Therefore, QUEKO circuits have known optimal depth $T$.

Note that QUEKO circuits also have known optimal gate count $M_1+M_2$. 
Since we assume that, in layout synthesis, all the input gates need to be executed, the result produced by the tools has at least as many gates as the QUEKO circuit. 
The optimal gate count is also met with the optimal initial mapping $\tau^{-1}$, since no SWAP gates are needed in this case.

\section{Experiment}\label{sec:experiment}
\subsection{Experimental Setup}

To evaluate QC layout synthesis tools with QUEKO, device graphs, depths and sizes, and gate density vectors are required. 
We specify the choice of these parameters and the choice of tools to evaluate in this subsection.
Because of the randomness in our construction, we generate ten QUEKO benchmarks for each triplet of the parameters.
These benchmarks are made open-source\footnote{\label{nt:queko}\url{https://github.com/UCLA-VAST/QUEKO-benchmark}} under the BSD license.
For evaluation, we fed each one of these benchmarks to four layout synthesis tools, as shown in Fig.~\ref{fig:flow}.
All the experiments were run on a Ubuntu 16.04 server, which has two Intel Xeon E5-2699v3 as CPUs and 128GB main memory.\footnote{We did not compare the runtime of these four tools, as some of them are implemented in Python and others in C. The inherent efficiency difference of programming languages renders a direct runtime comparison unfair. More fundamentally, since we use a noiseless model of quantum computers in this paper, layout synthesis would be required only once, so performance rather than runtime is the major concern. For other studies on mitigating changing noise properties of quantum devices, the runtime would be more important.}

\subsubsection{Device Graph}

We used representative devices from three different QC hardware providers. 
Sycamore from Google~\cite{arute-2019-quantum}, Tokyo and Rochester from IBM\footnote{\url{https://quantum-computing.ibm.com/}}, and Aspen-4 from Rigetti\footnote{\url{https://www.rigetti.com/qpu}}. 
The graphs of these devices are shown in Fig.~\ref{fig:layouts}. 
Sycamore has 54 qubits, of which 53 are active; Rochester also has 53 qubits. 
Both of them are state-of-the-art devices, but Sycamore has richer connectivity. 
Aspen-4 has 16 qubits, and Tokyo has 20 qubits. They are both highly competitive devices, but Tokyo has greater connectivity. 

Also, we have only listed superconducting devices because they are by far the most advanced QC devices. 
This does not mean that QUEKO cannot generalize to other technologies such as quantum dot\footnote{\url{https://sqc.com.au}} because our approach is valid as long as the basic quantum gates of this technology are single-qubit and two-qubit gates.

\subsubsection{Depth and Size} 
We constructed two sets of benchmarks with different depth ranges.
The corresponding size of these benchmarks can be deduced from the depth and the gate density vector, as shown in Algorithm \ref{alg:construction}.
The first set has depths from 5 to 45, which is the {\it near-term feasible benchmarks} (B\textsubscript{NTF}). 
In fact, one of the largest quantum circuits executed nowadays has depths 41~\cite{arute-2019-quantum}, which is about the same with the upper bound of B\textsubscript{NTF}. 
We intended to find out the layout synthesis performance within the current execution capacity.
The sizes of B\textsubscript{NTF} benchmarks range from 1136 to 34506 quantum gates.
The second set of benchmarks, denoted as B\textsubscript{SS} has depth from 100 to 900 which are {\it benchmarks for scaling study}. 
B\textsubscript{SS} represents the performance of these tools when the decoherence time of QC device improves in the future.
The sizes of B\textsubscript{SS} benchmarks range from 37 to 1727 quantum gates.

\subsubsection{Gate Density Vector}
We picked two special gate density vectors in the experiment: $(0.51, 0.4)$ based on the quantum circuits used in Google's quantum supremacy experiment~\cite{arute-2019-quantum}, denoted ``QSE'' below, and $(0.27, 0.36)$ based on the Toffoli circuit, denoted ``TFL'' below. 
It is beneficial to study QSE, since it is the only circuit so far with which experimental QC has shown a clear advantage. 
We chose the TFL because existing QC logic synthesis algorithms are based largely on reversible logic synthesis, which uses TFL as a fundamental element~\cite{shende-2003-synthesis}. 
Some B\textsubscript{NTF} have TFL density and others have QSE density.
All B\textsubscript{SS} have QSE density.
We also swept through possible gate density vectors and generated {\it benchmarks for impact of gate density} (B\textsubscript{IGD}).

\subsubsection{Layout Synthesis Tools}
\begin{figure}[htb]
    \centering
    \includegraphics[width=0.6\linewidth]{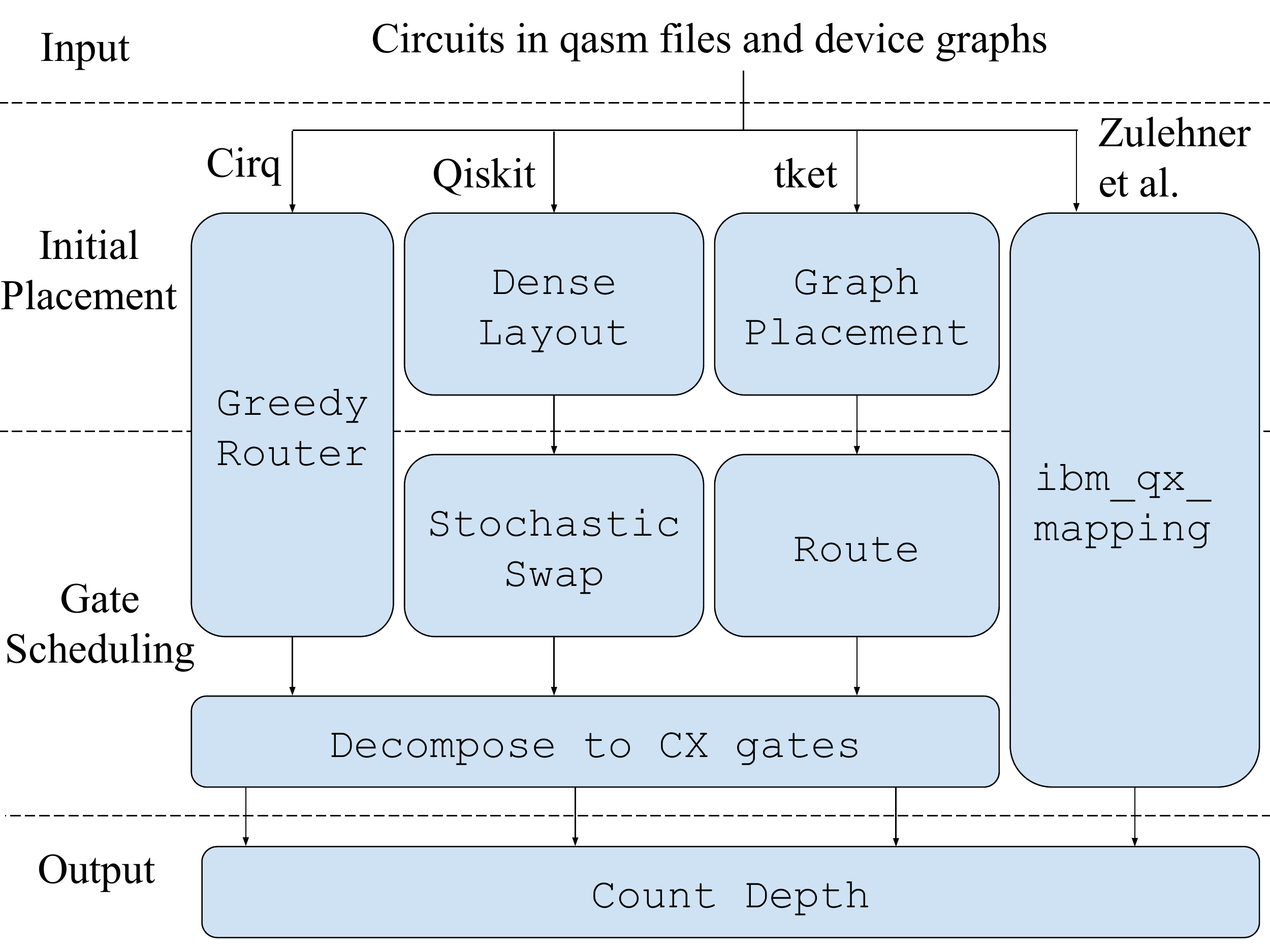}
    \caption{Workflow of the experiments}
    \label{fig:flow}
\end{figure}

Currently, Google, IBM, and Rigetti are considered front-runners of superconducting QC. 
Inside their QC frameworks (Cirq, Qiskit, and pyQuil), there are tools for layout synthesis. 
Unfortunately, pyQuil does not provide options to breakdown the whole compilation into optimization and layout synthesis, so pyQuil was excluded from the experiments. 
We also included a recent academic work from Zulehner et al.~\cite{zulehner-2018-efficient}, which is open-source.\footnote{\url{https://github.com/iic-jku/ibm_qx_mapping}}

We used \verb|greedy| router in Cirq version 0.6.0 as one of the layout synthesis tools, as shown in Fig.~\ref{fig:flow}.
So far, only one router named \verb|greedy| has been released, which contains an initial placement policy and a SWAP insertion policy based on heuristic search.
Note that \verb|greedy router| does not transform the gates into gates that are implementable on Google devices in Table \ref{tab:gates}, so the resulting circuit it produces contains the original input gates and SWAP gates it inserts.
For a fair comparison of depth, in all the experiments, we decompose SWAP gates inserted by the tools to three $CX$ gates, like in Fig.~\ref{fig:scheduling}.

Qiskit offers the most precise control over the so-called ``transpiler''. 
The transpilation is divided into individual passes, and users can define their own ``pass manager'' to make use of various transpiling modules that are offered. 
For the layout synthesis problem, there are \verb|Layout| modules generating initial mapping and \verb|Swap| modules inserting SWAP gates to the circuit to enable two-qubit gates.
Among the various combinations, we chose \verb|DenseLayout| and \verb|StochasticSwap| as shown in Fig.~\ref{fig:flow}, which seemed to have the best overall performance. 
\verb|DenseLayout| maps the logical qubits to an area on the device graph with dense connections.
\verb|StochasticSwap| perturbs the distance matrix of physical qubits and performs heuristic search for SWAP gates.
The version of Qiskit in the experiments is 0.14.1.

Another highly competitive router, $\mathsf{t}|\mathsf{ket}\rangle$, comes from Cambridge Quantum Computing. 
\verb|Graph Placement| uses graph monomorphism to derive initial mapping.
\verb|Route| performs heuristic search for SWAP gates.
We used $\mathsf{t}|\mathsf{ket}\rangle$ version 0.4.1 in the experiments.

Since all the tools evaluated use heuristics at some stages, sub-optimality is expected. 
Note that we only use default setup on all the modules in all the tools.
Changing setup parameters in some of these modules specifically for QUEKO benchmarks may lead to better performance in the following experiments, but may lead to worse performance on other circuits.

\subsection{Experimental Results}

\subsubsection{Performance on B\textsubscript{NTF}}

\begin{figure}[htb]
\centering
\subfloat[Smaller device (Aspen-4), sparser circuits (TFL)]{
\includegraphics[width=0.45\linewidth]{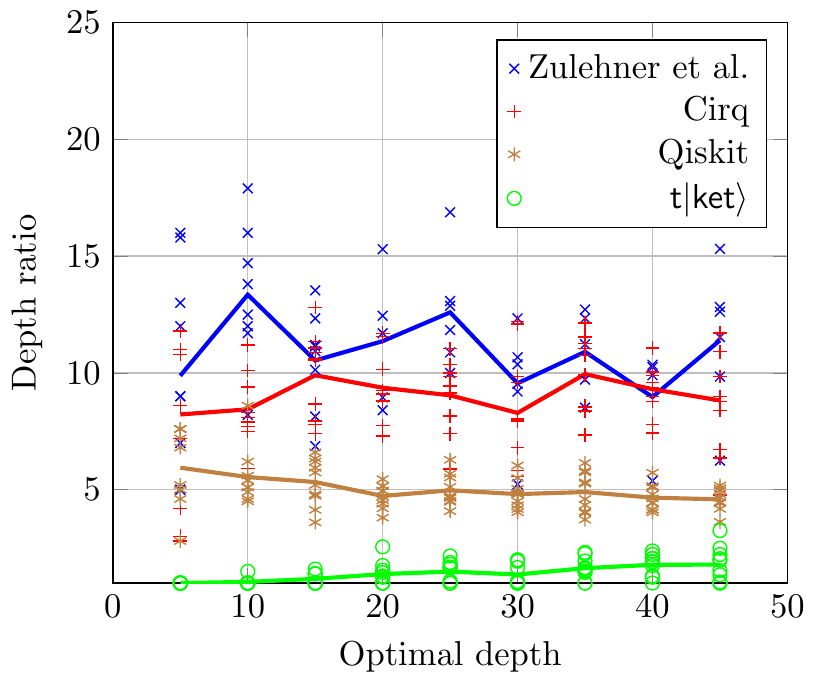}
\label{fig:aspen-rev}
}
\hfill
\subfloat[Larger device (Sycamore), denser circuits (QSE)]{
\includegraphics[width=0.45\linewidth]{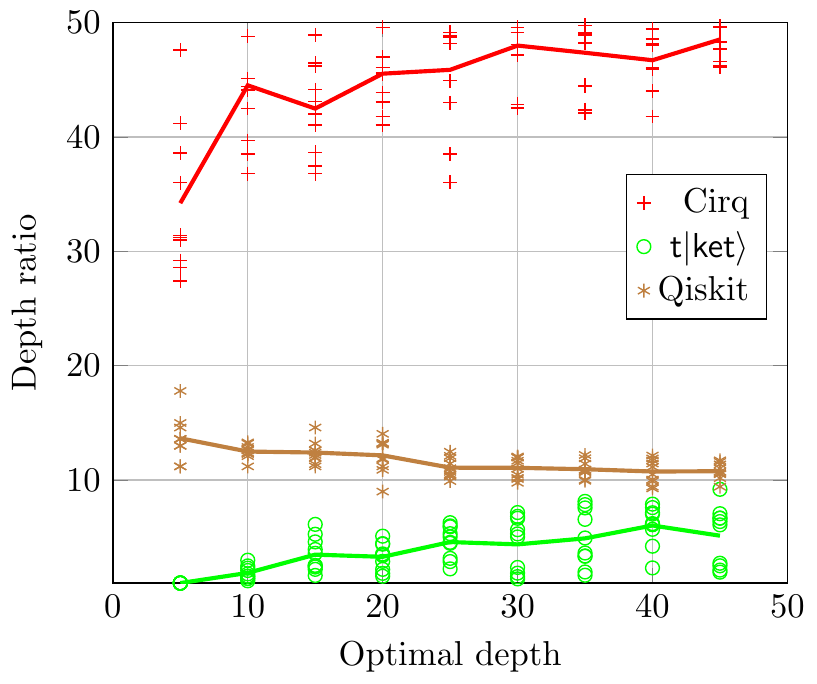}
\label{fig:Sycamore_Supremacy}
}
\caption{Performance of QC layout synthesis tools on B\textsubscript{NTF} (Lines are average.)}
\label{fig:PerformanceRealistic}
\end{figure}

In Fig.~\ref{fig:PerformanceRealistic}, the horizontal axis is the optimal depth and the vertical axis is the depth ratio, which is the depth of layout synthesis result divided by the optimal depth $T$. 
In the case of a smaller device (Aspen-4) and sparser circuits (TFL), the optimality gap on average is about 12x for~\cite{zulehner-2018-efficient}, 10x for Cirq, 5x for Qiskit and 1.5x for $\mathsf{t}|\mathsf{ket}\rangle$. 
In the case of a larger device (Sycamore) and denser circuits (QSE), the optimality gap on average is about 11x to 14x for Qiskit. 
The optimality gaps of Cirq and $\mathsf{t}|\mathsf{ket}\rangle$ grow with depth correspondingly from 35x to 50x and from 1x to 5x. 
Zulehner et al. is not in the Fig.~\ref{fig:Sycamore_Supremacy}, because for the larger device, it took so much memory that the operating system constantly killed it before finishing. 
This also happens sometimes for the smaller device experiments, so there are less blue data points than the other types of points in Fig.~\ref{fig:aspen-rev}.

\subsubsection{Performance on B\textsubscript{SS}}

\begin{figure}[htb]
\centering
\subfloat[$\mathsf{t}|\mathsf{ket}\rangle$ scaling behavior]{
\includegraphics[width=0.45\linewidth]{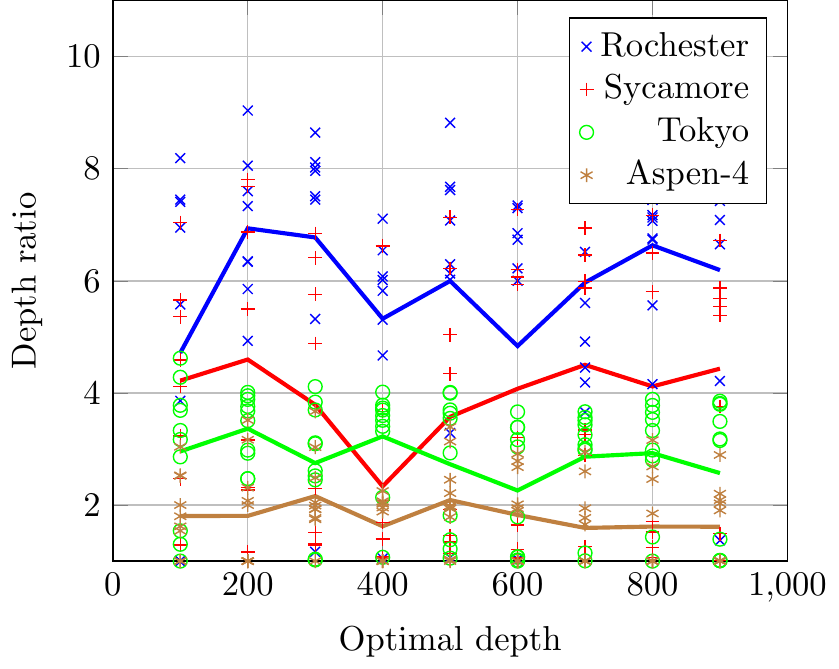}
\label{fig:pytket_scaling}
}
\hfill
\subfloat[Qiskit scaling behavior]{
\includegraphics[width=0.45\linewidth]{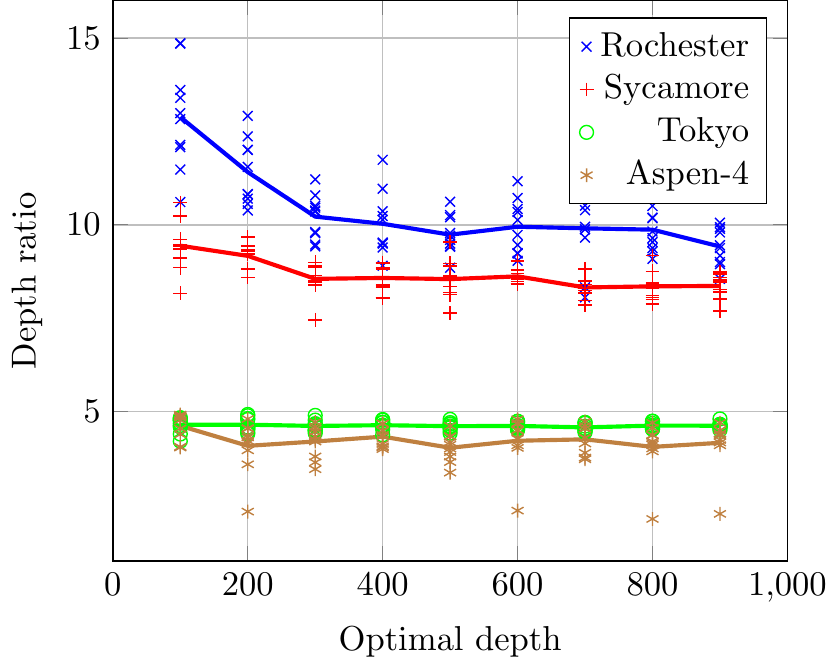}
\label{fig:qiskit_scaling}
}
\caption{$\mathsf{t}|\mathsf{ket}\rangle$ and Qiskit performance on B\textsubscript{SS} (Lines are average.)}
\label{fig:scaling}
\end{figure}

We studied further scaling of $\mathsf{t}|\mathsf{ket}\rangle$ and Qiskit on different devices as shown in Fig.~\ref{fig:scaling}.
The optimality gaps on average by $\mathsf{t}|\mathsf{ket}\rangle$ are about 5x to 7x for Rochester, 3x to 4x for Sycamore, 3x for Tokyo, and 2x for Aspen-4.
Note that for a fixed depth and a fixed device, the optimality gaps by $\mathsf{t}|\mathsf{ket}\rangle$ varies rather widely.
$\mathsf{t}|\mathsf{ket}\rangle$ managed to find the optimal mappings for some QUEKO benchmarks.
In general, as the depth increases, the depth ratio by Qiskit decreases at first and then converges to a value. 
The reason for this phenomenon may be that as the circuit deepens, the influence of initial placement gets smaller than the influence of SWAP insertion. 
For Qiskit, the optimality gaps on Rochester decreased from 13x to 10x on average; on Sycamore, Tokyo and Aspen-4 are about 8x, 5x, and 4x on average.
It can be seen that larger devices (Rochester and Sycamore) bring about larger optimality gaps. 
If the number of physical qubits are close, then richer connectivity (Sycamore versus Rochester) brings about smaller optimality gaps.

\subsubsection{Performance on B\textsubscript{IGD}}

\begin{figure}[htb]
\centering
\subfloat[$\mathsf{t}|\mathsf{ket}\rangle$ performance in depth ratio]{
\includegraphics[width=0.45\linewidth]{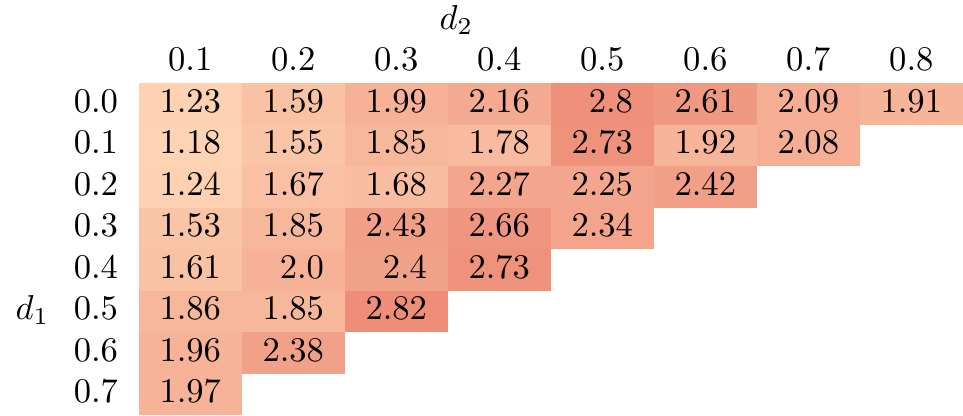}
\label{fig:pytket_heatmap}
}
\hfill
\subfloat[Qiskit performance in depth ratio]{
\includegraphics[width=0.45\linewidth]{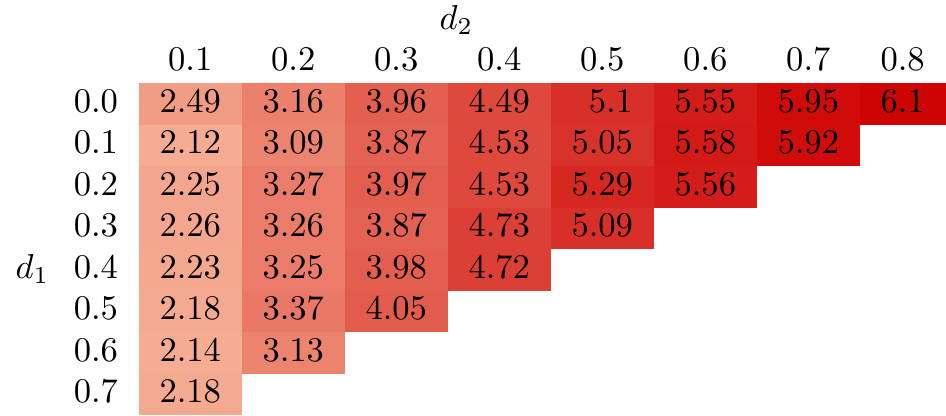}
\label{fig:qiskit_heatmap}
}
\caption{$\mathsf{t}|\mathsf{ket}\rangle$ and Qiskit performance on B\textsubscript{IGD} (Data are 10-time average.)}
\label{fig:gatedist}
\end{figure}

To better understand the impact of gate density on layout synthesis performance, we fixed the device to Tokyo and the depth to 45, and swept through possible gate densities. 
The results are shown in Fig.~\ref{fig:gatedist}. 
Fixing a column, the single-qubit gate density increases as we go down, Qiskit seems to be rather insensitive to this change, which is sensible since the single-qubit gates do not induce difficulty in layout synthesis. 
However, $\mathsf{t}|\mathsf{ket}\rangle$ is still sensitive to this change. 
Both tools are more sensitive to the change in the horizontal direction than in the vertical direction. 
Since the challenge to layout synthesis comes mainly, if not solely, from the two-qubit gates, this result is expected.
The depth ratio of $\mathsf{t}|\mathsf{ket}\rangle$ decreases when the two-qubit gate density is very high.
This is because when the circuit is dense with two-qubit gates, graph monomorphism algorithms can extract more information from the first few layers to narrow down better initial mappings.

\section{Complexity}
\label{sec:complexity}
Seeing the large optimality gap, it is natural for us to investigate the computational complexity of the depth-optimal QC layout synthesis problem, which was unknown till this point. 
Several related results are shown, e.g., determining the minimal number of SWAP gates to insert is NP-complete.
\cite{siraichi-2018-qubit} proves this theorem by reduction from subgraph isomorphism problem.
\cite{maslov-2008-quantum} proves the NP-completeness of depth-optimal initial placement where SWAP gates are not allowed.
The NP-completeness of depth-optimal QC layout synthesis for QAOA circuits is proven in~\cite{botea-2018-complexity} by reduction from 3-SAT. 
In this section, we prove this for general quantum circuits by reduction from Hamiltonian cycle problem, as Theorem~\ref{thm} states.

\begin{theorem} \label{thm}
Depth-optimal QC layout synthesis is NP-complete.
\end{theorem}

\begin{proof}
The original QC depth-optimal layout synthesis problem is not easier than its decision version: input and constraints remain the same; but output whether the depth of the scheduled circuit can be lower or equal to $T$.
Inspired by~\cite{maslov-2008-quantum}, we show that the problem of determining whether a Hamiltonian cycle exists in a graph is reducible to the QC depth-decision layout synthesis problem. 
The former is NP-complete, so the latter is also NP-complete.

Suppose the graph for the Hamiltonian cycle problem is $G_H=(V_H, E_H)$, where $|V_H|=N$. 
We construct a depth-decision QC layout synthesis problem using $G_H$ as the device graph and $N$ as the target depth. 
The input circuit is $N$ ``levels'' of gates. 
Level $l$ contains a two-qubit gate $g_l=CX(q_l, q_{(l+1) \operatorname{mod} N})$. 
All the other logical qubits at level $l$ are fully occupied by single-qubit gates.

If there exists a Hamiltonian cycle in $G_H$, $(p_1, p_2,..., p_N,p_1)$, then let the initial mapping be $\mu_0: q_i\mapsto p_i$ for $i=1$ to $N$. 
It is easy to see that, with this mapping, all the gates in the constructed circuit can be executed. 

On the other hand, if there exists a scheduled circuit with depth within $N$, we first claim that the mapping cannot change during the execution of the circuit. 
Every gate in a level depends on some gate in the last level. 
So every gate in level $l$ has a dependency chain of length $l$, which is the earliest cycle it can be scheduled to. 
This means, if any SWAP gates are inserted in gate scheduling, certain dependency chain must lengthen and the depth of the scheduled circuit is larger than $N$. 
Therefore, if a solution within $N$ cycles exists, each gate in level $l$ must be scheduled at exactly cycle $l$ and there can be no SWAP gates inserted. 
It is also easy to see that the input $CX$ gates cannot constitute any SWAP gates. 
Therefore, the mapping from logical to physical qubits remains $\mu_0$ throughout all the cycles in the scheduled circuit. 
The gates $CX(q_l, q_{(l+1) \operatorname{mod} N})$ being all executable means that they are mapped to edges of $G_H$. 
This means $\mu_0(q_1)$, $\mu_0(q_2)$,..., $\mu_0(q_N)$, $\mu_0(q_1)$ is a Hamiltonian cycle in $G_H$.

In conclusion, we established the equivalence between a Hamiltonian cycle in $G_H$ and the existence of QC layout synthesis solution within depth $N$. Thus, the Hamiltonian cycle problem is reducible to the QC layout synthesis problem. The latter is NP-complete, since the former is known to be NP-complete.
\end{proof}

\section{Conclusion and Future Work} \label{sec:conclusion}

In this paper, we formulated the problem of quantum computing layout synthesis and proved its NP-completeness. 
We constructed QUEKO benchmarks, each has a known optimal depth for the given device.
With QUEKO, we examined four existing quantum computing layout synthesis tools, \verb|greedy| router in Cirq, $\mathsf{t}|\mathsf{ket}\rangle$, \verb|DenseLayout| plus \verb|StochasticSwap| in Qiskit, and Zulehner et al.~\cite{zulehner-2018-efficient} and showed rather surprising results. 
Despite over ten years of research and development efforts by both academia and industry, the current QC compilation flow is far from optimal. 
In fact, even combining the best performances of all tools evaluated, the optimality gaps range from 1.5x to 45x on average for circuits of feasible depth on existing devices. 
These gaps reveal that there is substantial room for research into QC layout synthesis, potentially equivalent to an order of improvement of the decoherence time, which would require much higher investment in quantum device technologies to achieve.
We plan to use QUEKO benchmarks as a guide to better layout synthesis tools.
On the other hand, our construction method opens up room for future extension.
With more specific control over the amount and structures of dependency chains, we believe it is possible to generate benchmarks with known optimal that are even more representative of real applications, which sometimes requires SWAP gates in optimal solutions.
In addition, we plan to extend our research to construct examples with known optimal solutions for fidelity optimization, as multiple studies have shown that fidelity is a very important metric for quantum circuits in the NISQ era.

\section*{Acknowledgment}
The authors would like to thank Iris Cong, Silas Dilkes, Dmitri Maslov, and Nengkun Yu for valuable comments on the manuscript.
This work is partially supported by NEC under the Center for Domain-Specific Computing Industrial Partnership Program.

\bibliographystyle{unsrt}
\bibliography{BochenTan}

\end{document}